%% file: e02_brzozowski.tex
\documentclass[copyright]{eptcs}
\providecommand{\event}{J.~Dassow, G.~Pighizzini, B.~Truthe (Eds.): 11th International Workshop\\
on Descriptional Complexity of Formal Systems (DCFS 2009)} 
\providecommand{\volume}{3}
\providecommand{\anno}{2009}
\providecommand{\firstpage}{17} 
\providecommand{\eid}{2}
\usepackage{breakurl}        

\input{dcfs.tex}

\usepackage{graphicx}
\usepackage{epic}
\usepackage{eepic}

\usepackage{algorithm}
\usepackage{algorithmic}

\newcommand{\ol}{\overline}
\newcommand{\eps}{\varepsilon}
\newcommand{\emp}{\emptyset}
\newcommand{\Sig}{\Sigma}

\newcommand{\ur}{uniquely reachable}
\newcommand{\bi}{\begin{itemize}}
\newcommand{\ei}{\end{itemize}}
\newcommand{\be}{\begin{enumerate}\rm}
\newcommand{\ee}{\end{enumerate}}
\newcommand{\bd}{\begin{description}}
\newcommand{\ed}{\end{description}}
\newcommand{\bq}{\begin{quote}}
\newcommand{\eq}{\end{quote}}
\newcommand{\txt}[1]{\mbox{ #1 }}
\newcommand{\qedb}{$\diamond$}

\newcommand{\etc}{\mbox{\it etc.}}
\newcommand{\ie}{\mbox{\it i.\,e.}}
\newcommand{\eg}{\mbox{\it e.\,g.}}

\newcommand{\cA}{{\mathcal A}}
\newcommand{\cB}{{\mathcal B}}

\newcommand{\cL}{{\mathcal L}}

\begin{document}
\title{Quotient Complexity of Regular Languages\,%
\thanks{This research was supported by the Natural Sciences and Engineering Research Council of Canada under grant no.\ OGP0000871.}
}
\def\titlerunning{Quotient Complexity of Regular Languages}
\def\authorrunning{J.~Brzozowski}

\author{Janusz Brzozowski
\institute{David R.~Cheriton School of Computer Science\\
 University of Waterloo\\
 Waterloo, ON --
 Canada N2L 3G1}
\email{brzozo@uwaterloo.ca}
}

\maketitle

\begin{abstract}
The past research on the state complexity of operations on regular languages is examined, and
 a new approach based on an old method (derivatives of regular expressions) is presented. Since state complexity is a property of a language, it is appropriate to define it in formal-language terms as the number of distinct quotients of the language, and to call it ``quotient complexity''.
The problem of finding the quotient complexity of a language $f(K,L)$ is considered, where $K$ and $L$ are regular languages and $f$ is a regular operation, for example, union or concatenation.
Since quotients can be represented by derivatives, one can find a formula for the typical quotient of $f(K,L)$ in terms of the quotients of $K$ and $L$. To obtain an upper bound on the number of quotients of $f(K,L)$  all one has to do is count how many such quotients are possible, and this makes automaton constructions unnecessary. The advantages of this point of view are illustrated by many examples. Moreover, new general observations are presented to help in the estimation of the upper bounds on quotient complexity of regular operations.
\end{abstract}

\section{Introduction}
\label{sec:intro}
It is assumed that the reader is familiar with the basic concepts of regular languages and finite automata, as described in many textbooks. 
General background material can be found in Dominique Perrin's~\cite{Per90} (1990) and Sheng Yu's~\cite{Yu97} (1997) handbook articles; the latter has an introduction to state complexity. A~more detailed treatment of state complexity can be found in Sheng Yu's survey~\cite{Yu01}. The present paper concentrates on the complexity of basic operations on regular languages. Other aspects  of complexity of regular languages and finite automata are discussed  in~\cite{BHK08,GrHo08,HoKu09,Jir209,JiOk08,JiPi09,SSY07,SaYu07}; this list is not exhaustive, but it should give the reader a good idea of the scope of the work on this topic.

\section{State complexity or quotient complexity?}
\label{sec:quotient}
The English term \emph{state complexity} of a regular language seems to have been introduced by Birget\footnote{An error in~\cite{Bir91} was corrected in~\cite{YuZh91}.}~\cite{Bir91} in 1991, and is now in common use. It is defined as the number of states in the minimal deterministic finite automaton (DFA) accepting the language~\cite{Yu01}. There had been much earlier studies of  this topic, but the term ``state complexity'' was not used. For example, in 1963 Lupanov~\cite{Lup63} showed that the bound~$2^n$ is tight for the conversion of nondeterministic finite automata (NFA's) to DFA's, and he used the term \emph{slozhnost' avtomatov,} meaning \emph{complexity of automata} representing the same set of words.
The case of languages over a one-letter alphabet was studied in 1964 by Lyubich~\cite{Lyu64}.
Lupanov's result is almost unknown in the English-language literature, and is often attributed to the 1971 paper by Moore~\cite{Moo71}. In~1970, Maslov~\cite{Mas70} studied the complexity of basic operations on regular languages, and  stated without proof some tight bounds for these operations. In the introduction to his paper he states:
\bq
An important characteristic of the complexity of these sets [of words] is the \emph{number of states of the minimal representing automaton}.\footnote{The emphasis is mine.}
\eq
In 1981 Leiss~\cite{Lei81}  referred to \emph{(deterministic) complexity} of languages.
Some additional references to early works related to this topic can be found 
in~\cite{HJS05,Yu01}, for example.

A \emph{language} is a subset of the free monoid $\Sig^*$ generated by a finite alphabet~$\Sig$. If state complexity is a property of a language, then why is it defined in terms of a completely different object, namely an automaton? Admittedly, regular languages and finite automata are  closely related, but there is a more natural way to define this complexity of languages, as is shown below.

The \emph{left quotient}, or simply \emph{quotient} of a language $L$ by a word $w$ is  defined 
as the language $$w^{-1}L=\{x\in \Sig^*\mid wx\in L \}.$$ 
The \emph{quotient complexity} of $L$ is the number of distinct languages that are quotients of $L$, and will be denoted by $\kappa(L)$ (\emph{kappa} for both \emph{kwotient} and \emph{komplexity}).
Quotient complexity is defined for \emph{any} language, and so may be finite or infinite.

Since languages are sets, it is natural to define set operations on them. The following are typical set operations:  {\em complement\/} ($\ol{L}=\Sig^*\setminus L$),  {\em union\/}  ($K\cup L$),  {\em intersection\/} ($K\cap L$),  {\em difference\/} ($K\setminus L$), and {\em symmetric difference\/} ($K\oplus L$). A general {\em boolean operation\/} with two arguments is denoted by \hbox{$K\circ L$}. 
Since languages are also subsets of a monoid, it is also natural to define \emph{product}, usually called  \emph{(con)catenation},   ($K\cdot L=\{w\in \Sig^*\mid w=uv, u\in K, v\in L\}$), 
  \emph{star} ($K^*=\bigcup_{i\ge 0}K^i$), and \emph{positive closure} ($K^+=\bigcup_{i\ge 1}K^i$).

The operations union, product and star are called \emph{rational} or \emph{regular}. \emph{Rational (or regular) languages} over $\Sig$ are those languages that can be obtained from the set $\{\emp, \{\eps\}\}\cup \{\{a\}\mid a\in \Sig\}$ of  \emph{basic languages}, where $\eps$ is the empty word,  (or, equivalently, from another basis, such as the finite languages over $\Sigma$) using a~finite number of rational operations. Since it is cumbersome to describe regular languages as sets---for example, one has to write $L=(\{\eps\}\cup\{a\})^*\cdot\{b\}$---one normally switches to regular (or rational) expressions. These are the terms of the free algebra over the set $\Sigma \cup\{\emp,\eps\}$ with function symbols\footnote{The symbol $+$ is used instead of $\cup$ in~\cite{Per90}.} $\cup$, $\cdot$, and $^*$~\cite{Per90}. For the example above, one  writes $E=(\eps\cup a)^*\cdot b$.
The mapping $\cL$ from this free algebra onto the algebra of regular languages is defined inductively as follows:
$$\cL(\emp)=\emp, \quad \cL(\eps)=\{\eps\}, \quad \cL(a)=\{a\},$$ 
$$\cL(E\cup F) =\cL(E)\cup \cL(F), \quad \cL(E\cdot F) =\cL(E)\cdot \cL(F), \quad\cL(E^*)=(\cL(E))^*,$$
 where $E$ and $F$ are regular expressions.
The product symbol $\cdot$ is usually dropped, and  languages are denoted by expressions without further mention of the mapping~$\cL$. Since regular languages are closed under complementation, complementation is treated here as a regular operator.

Because regular languages are defined by regular expressions, it is natural to  use regular expressions also to represent their quotients; these expressions are their derivatives~\cite{Brz64}. First, the \emph{$\eps$-function}  of a regular expression $L$, denoted by~$L^\eps$, is defined as follows:
\begin{eqnarray}
a^\eps &=& 
	\left\{ \begin{array}{ll}
		\emp, & \mbox{if $a=\emp$, or $a\in\Sig$};\\
		\eps, & \mbox{if $a=\eps$}.
	\end{array}
	\right.
\end{eqnarray}

\begin{eqnarray}
(\ol{L})^\eps &=& 
	\left\{ \begin{array}{ll}
		\emp, & \mbox{if $L^\eps=\eps$};\\
		\eps, & \mbox{if $L^\eps=\emp$}.
	\end{array}
	\right.
\end{eqnarray}

\begin{equation}
(K\cup L)^\eps = K^\eps\cup L^\eps, \:
(KL)^\eps = K^\eps \cap L^\eps, \:
(L^*)^\eps= \eps.
\end{equation}
One verifies that $\cL(L^\eps)=\{\eps\}$ if $\eps\in L$, and $\cL(L^\eps)=\emp$, otherwise.

The \emph{derivative by a letter} $a\in\Sig$ of a regular expression $L$ is denoted by $L_a$ and defined by structural induction:
\begin{eqnarray}
b_a &=& 
	\left\{ \begin{array}{ll}
		\emp, & \mbox{if $b\in \{\emp,\eps\}$, or $b\in\Sig$ and $b\not= a$};\\
		\eps, & \mbox{if $b=a$}.
	\end{array}
	\right.
\end{eqnarray}
\begin{equation}
\label{eq:derlet}
(\ol{L})_a =\ol{L_a}, \:
(K\cup L)_a = K_a\cup L_a, \:
(KL)_a = K_aL \cup K^\eps L_a, \:
(L^*)_a = L_aL^*.
\end{equation}
The \emph{derivative by a word} $w\in\Sig^*$ of a regular expression $L$ is denoted by $L_w$ and defined by
induction on the length of $w$:
\begin{equation}
L_\eps =  L, \quad
L_w = L_a, \: \mbox{if $w=a\in \Sig$}, \quad
L_{wa} = (L_w)_a.
\end{equation}
A derivative $L_w$ is {\em accepting\/} if $\eps\in L_w$; otherwise it is {\em rejecting\/}.

One can verify by structural induction that $\cL(L_a)=a^{-1}L, \txt{for all} a\in\Sigma$, and then by induction on the length of $w$ that, for all $w\in\Sigma^*$,
\begin{equation}
\cL(L_w)=w^{-1}L.
\end{equation}
Thus every derivative represents a unique quotient of $L$, but there may be many derivatives representing the same quotient.

Two regular expressions are \emph{similar}~\cite{Brz62,Brz64} if one can be obtained from the other using the following rules:
\begin{equation}
L\cup L = L, \quad K\cup L=L\cup K, \quad K\cup (L\cup M)= (K\cup L)\cup M,
\end{equation}
\begin{equation}
L\cup \emp= L, \quad \emp L= L\emp =\emp, \quad \eps L= L\eps =L.
\end{equation}

Upper bounds on the number of dissimilar derivatives, and hence on the quotient complexity, were derived in~\cite{Brz62,Brz64}: If $m$ and $n$ are the quotient complexities of $K$ and $L$, respectively, then
\begin{equation}
\label{eq:JABbounds}
 \kappa(\ol{L})=\kappa(L), \quad \kappa(K\cup L)\le mn, \quad \kappa(KL)\le m2^n, \quad 
\kappa(L^*)\le 2^n -1.
\end{equation}
This immediately implies that the number of derivatives, and hence the number of quotients, of a regular language is finite.

It seems that the upper bounds in Equation~(\ref{eq:JABbounds}), derived in 1962~\cite{Brz62,Brz64}, were the first ``state complexity" bounds to be found for the regular operations. Since the aim at that time was simply to show that the number of quotients of a regular language is finite, the tightness of the bounds was not considered. 

Of course, the concepts above are related  to the more commonly used ideas.
A~\emph{deterministic finite automaton}, or simply \emph{automaton,} is a tuple $$\cA=(Q, \Sig, \delta, q_0,F),$$ where
$Q$ is a finite, non-empty set of \emph{states}, $\Sig$ is a finite, non-empty \emph{alphabet}, $\delta:Q\times \Sig\to Q$ is the \emph{transition function}, $q_0\in Q$ is the \emph{initial state}, and $F\subseteq Q$ is the set of \emph{final states}.
The transition function is extended to $\delta:Q\times \Sig^*\to Q$ as usual. A word $w$ is \emph{recognized} (or \emph{accepted}) by automaton $\cA$ if $\delta(q_0,w)\in F$.
It was proved by Nerode~\cite{Ner58} that a language $L$ is recognizable by a finite automaton if and only if $L$ has a finite number of quotients.

The \emph{quotient automaton} of a regular language $L$ is 
$\cA=(Q, \Sig, \delta, q_0,F)$, where $Q=\{w^{-1}L\mid w\in\Sig^*\}$, $\delta(w^{-1}L,a)=(wa)^{-1}L$, 
$q_0=\eps^{-1}L=L$, and \hbox{$F=\{w^{-1}L \mid \eps\in w^{-1}L\}$}.

It should now be clear that the state complexity of a regular language $L$ is the number of states in its quotient automaton, \ie, the number $\kappa(L)$  of its quotients.
This terminology change may seem trivial, but has some nontrivial  consequences.

For convenience, derivative notation will be used to represent quotients, in the same way as regular expressions are used to represent regular languages. 

By convention, $L_w^\eps$  always means $(L_w)^\eps$. 

Several proofs are omitted because of space limitations.

\section{Derivation of  bounds using quotients}
\label{sec:bounds}
Since languages over one-letter alphabets have very special properties, we usually assume that the alphabet has at least two letters.
The complexity of operations on unary languages has been studied in~\cite{PiSh02,Yu01}.

In the  literature on state complexity, it is assumed that  automata $\cA$ and $\cB$ accepting languages $K$ and $L$, respectively, are given. An assumption has to be made that the automata are ``complete'', \ie, that for each
$q\in Q$ and $a\in\Sig$, $\delta(q,a)$ is defined~\cite{YZS94}. 
In particular,  if a ``dead" or ``sink" state, which accepts no words, is present,  one has to check that only one such state is included~\cite{HaSa09}. Also, every state must be ``useful'' in the sense that it appears on some accepting path~\cite{HSW06}. 

Suppose that a bound on the state complexity of $f(K,L)$ is to be computed, where $f$ is some regular operation. 
In some cases a DFA accepting $f(K,L)$ is constructed directly, 
(\eg, Theorems~2.3 and~3.1 in~\cite{YZS94}), 
or an NFA  with multiple  initial states is used,  
and then converted to a DFA by the subset construction (\eg, Theorem 4.1 in~\cite{YZS94}).
Sometimes an NFA with empty-word transitions is used and then converted to a DFA~\cite{SaYu07}.
The constructed  automata  then have to be proved minimal.

Much of this is unnecessary. If quotients are used, the problem of completeness does not arise, since all the quotients of a language are included. A quotient  is either empty or  ``useful''. If the empty quotient is present, then it appears only once. Since quotients are distinct languages, the set of quotients of a language is always minimal.
To find an upper bound on the state complexity, instead of constructing an automaton for $f(K,L)$, we need only find a regular expression for the typical quotient, and then do some counting. This is illustrated below for the basic regular operations.

\subsection{Bounds for basic operations}
\label{subsec:basic}
The following are some useful formulas for  the derivatives of regular expressions:
\begin{theorem}
\label{thm:der}
If $K$ and $L$ are regular expressions, then

\begin{equation}
\label{eq:comp}
(\ol{L})_w=\ol{L_w},
\end{equation}

\begin{equation}
\label{eq:bool}
(K\circ L)_w=K_w\circ L_w,
\end{equation}
\begin{equation}
\label{eq:prod}
(KL)_w=  K_wL \cup K^\eps L_w\cup\left(\bigcup_{{w=uv}\atop {\;\; u,v\in\Sigma^+}} 
K_u^\eps L_v\right).
\end{equation}
For the Kleene star, $(L^*)_\eps=\eps\cup LL^*$, and for  $w\in\Sigma^+$, 
\begin{equation}
\label{eq:star}
 (L^*)_w=  \left( 
\bigcup_{{w=uv}\atop{u,v\in\Sigma^*}}
(L^*)_u^\eps
L_v
\right)
L^*.
\end{equation}
\end{theorem}

Theorem~\ref{thm:der} can be applied to obtain upper bounds on the complexity of  operations. 
In Theorem~\ref{thm:basic} below, the second part is a slight generalization of the  bound in Theorem~4.3 of~\cite{YZS94}.
The third and fourth parts are  reformulations of  the bounds in Theorem~2.3 and~2.4, and 
 of Theorem~3.1 of~\cite{YZS94}:
\begin{theorem}
\label{thm:basic}
For any languages $K$ and $L$ with $\kappa(K)=m$ and $\kappa(L)=n$:
\be
\item
$
\kappa(\ol{L})= n.
$
\item
$\kappa(K\circ L)\le mn.
$
\item{\it
Suppose $K$ has $k$ accepting quotients and $L$ has $l$ accepting quotients.}
	\be
	\item{\it
	 If $k=0$ or $l=0$, then $\kappa(KL)=1$.}
	\item{\it
	If  $k,l>0$ and $n=1$, then $\kappa(KL)\le m-(k-1)$.}
	\item{\it
	If $k,l>0$ and $n>1$, then $\kappa(KL)\le m2^n-k2^{n-1}. $}
	\ee
\item
	\be
	\item{\it
	If $n=1$, then $\kappa(L^*)\le 2$.}
	\item{\it
	If $n>1$ and $L_\eps$ is the only accepting quotient of  $L$, then
	$\kappa(L^*)=n$.}
	\item{\it
	If $n>1$ and   $L$ has $l>0$ accepting quotients not equal to $L$, then
	$\kappa(L^*)\le 2^{n-1}+2^{n-l-1}$.}
	\ee
\ee
\end{theorem}
\begin{proof}
The first part is well-known,  and the second follows from (\ref{eq:bool}).

For the product,
if $k=0$ or $l=0$, then $KL=\emp$ and $\kappa(KL)=1$. Thus assume that $k,l>0$.
If $n=1$, then $L=\Sig^*$ and  $w\in K$ implies $(KL)_w=\Sig^*$. 
Thus all $k$ accepting quotients of $K$ produce the one quotient  $\Sig^*$ in $KL$.
For each rejecting quotient of $K$, we have two choices for the union of quotients of~$L$ in~(\ref{eq:prod}): the empty union or $\Sig^*$.
If we choose the empty union, we can have at most $m-k$ quotients of $KL$.
Choosing $\Sig^*$ results in $(KL)_w=\Sig^*$, which has been counted already. Altogether, there are at most $1+m-k$ quotients of $KL$.
Suppose now that $k,l>0$ and $n>1$. If $w\notin K$, then we can choose~$K_w$ in $m-k$ ways, and the union of quotients of $L$ in $2^n$ ways. 
If $w\in K$, then we can choose~$K_w$ in $k$ ways, and the set of quotients of $L$ in $2^{n-1}$ ways, since $L$ is then always present.  Thus we have $(m-k)2^n+k2^{n-1}$.

For the star, if $n=1$, then $L=\emp$ or $L=\Sig^*$. In the first case, $L^*=\eps$, and $\kappa(L^*)=2$; in the second case, $L^*=\Sig^*$ and $\kappa(L^*)=1$.  
Now suppose that $n>1$; hence $L$ has at least one accepting quotient. If  $L$ is the only accepting quotient of $L$, then $L^*=L$ and $\kappa(L^*)=\kappa(L)$.

Now assume that $n>1$ and $l>0$.
From~(\ref{eq:star}), every  quotient of $L^*$ by a non-empty word is a union 
of a subset of quotients of $L$, followed by $L^*$. 
Moreover, that union is non-empty, because $(L^*)_\eps^\eps L_w$ is always present.
We have two cases:
\be
\item Suppose $L$ is rejecting. Then $L$ has $l$ accepting quotients.
	\be
	\item
	If no accepting quotient of $L$ is included in the subset, then there are $2^{n-l}-1$ such 	subsets possible, the union being non-empty because  $L_w$ is always included.
	\item
	If an accepting quotient of $L$ is included,
	then $\eps\in (L^*)_w$,  $(L^*)_w^{\eps}=\eps$,
	and  $L = (L^*)_w^{\eps}L_{\eps}$ is also included.
	We have $2^l-1$ non-empty subsets of accepting quotients of $L$ and $2^{n-l-1}$ subsets of rejecting quotients, since $L$ is not counted.
		\ee
		Adding 1 for $(L^*)_{\eps}$, we have a total of  
	$2^{n-l}-1+(2^l-1)2^{n-l-1}+ 1=2^{n-1}+2^{n-l-1}.$
\item Suppose $L$ is accepting. Then $L$ has $l+1$ accepting quotients.
	\be
	\item
	If there is no accepting quotient,  there are $2^{n-l-1}-1$ 	non-empty subsets of rejecting quotients.
	\item
	If an accepting quotient of $L$ is included, then $L$ is included, and  $2^{n-1}$ subsets can be added to $L$.
	\ee
We need not add  $(L^*)_\eps$, since $\epsilon \cup LL^*= LL^*$ in this case, and this has already been counted. The total is $2^{n-1}+2^{n-l}-1$.
\ee
The worst-case bound of $2^{n-1}+2^{n-l-1}$ occurs in the first case only.
\end{proof}
\goodbreak

\subsection{Witnesses to bounds for basic operations}
\label{subsec:witness}
Finding witness languages showing that a bound is tight is often challenging.
However, once a guess is made, the verification can be done using quotients.

Let $|w|_a$ be the number of $a$'s in $w$, for $a\in \Sig$ and $w\in \Sig^*$.
Unary, binary,  and ternary languages are languages over a one-, two-, and three-letter alphabet, respectively. 

\bi
\item
{\bf Union and Intersection}
If we have a bound for intersection, then for union we can use the fact that $\kappa(\ol{K}\cup \ol{L})=\kappa(\ol{\ol{K}\cup\ol{L}})=\kappa(K\cap L)$; thus the pair $(\ol{K},\ol{L})$ is a witness for union. Similarly, given a witness for union, we also have a witness for intersection.

The upper bound $mn$ for the complexity of  intersection was observed in 1957\footnote{The work was done in 1957, but published in 1959.} by Rabin and Scott~\cite{RaSc59}.
Binary languages  
$$K=\{w\in\{a,b\}^*\mid |w|_a\equiv m-1 \mbox{ mod } m\}$$ 
and 
$$L=\{w\in\{a,b\}^*\mid |w|_b\equiv n-1 \mbox{ mod } n\}$$
have quotient complexities $m$ and $n$, respectively. In 1970 Maslov~\cite{Mas70} stated without proof that $K\cup L$ meets this upper bound $mn$. 
Yu, Zhuang and K.~Salomaa~\cite{YZS94},  used similar languages 
$$K'=\{w\in\{a,b\}^*\mid |w|_a\equiv 0 \mbox{ mod } m\}$$ and 
$$L'=\{w\in\{a,b\}^*\mid |w|_b\equiv 0 \mbox{ mod } n\}$$
for intersection, apparently unaware of~\cite{Mas70}.
Hricko, Jir\'askov\'a and Szabari~\cite{HJS05} showed that a complete hierarchy of quotient complexities of binary languages exists between the minimum complexity 1 and the maximum complexity $mn$. 
More specifically, it was proved that for any integers $m,n,\alpha$ such that $m\ge2$, $n\ge 2$ and $1\le \alpha \le mn$, there exist binary\footnote{The proof in~\cite{HJS05} is for ternary languages; a proof for the binary case can be found in~\cite{Hri05}.} languages
 $K$ and $L$ such that $\kappa(K)=m$, $\kappa(L)=n$, and $\kappa(K\cup L)=\alpha$, and the same holds for intersection. 

For a one-letter alphabet $\Sig=\{a\}$, Yu showed that the bound can be reached if $m$ and $n$ are relatively prime~\cite{Yu01}. The witnesses are \hbox{$K''=(a^m)^*$} and $L''=(a^n)^*$.
For other cases, see the paper by Pighizzini and Shallit~\cite{PiSh02}.

\item
{\bf Set difference}
For set difference we have $\kappa(K'\setminus \ol{L'})=\kappa(K'\cap L')$; thus the pair $(K',\ol{L'})$ is a witness.

\item
{\bf Symmetric difference}
For symmetric difference, let $m,n\ge 1$, let \hbox{$K=(b^*a)^{m-1}(a\cup b)^*$} and let $L=(a^*b)^{n-1}(a\cup b)^*$.
There are $mn$ words of the form $a^ib^j$, where $0\le i\le m-1$ and \hbox{$0\le j\le n-1$}.
We claim that all the quotients of $K\oplus L$ by these words are distinct. Let $x=a^ib^j$ and $y=a^kb^l$.
If $i<k$, let $u=a^{m-1-k}b^n$. Then $xu\notin K$,  $yu\in K$, and $xu,yu\in L$, showing that 
$xu\in K \oplus L$, and $yu\notin K\oplus L$, \ie, that $(K\oplus L)_x\not= (K\oplus L)_y$.
Similarly, if $j<l$, let $v=a^mb^{n-1-l}$. Then $xv\in K\oplus L$, but $yv\notin K\oplus L$.
Therefore all the quotients of $K\oplus L$ by these $mn$ words are distinct.

For a one-letter alphabet, the witnesses are $K''$ and $L''$ as in the case of union above.
\item
{\bf Other boolean functions}
There are six more two-variable boolean functions  that depend on both variables: $\ol{K}\cup\ol{L}=\ol{K\cap L}$, $\ol{K}\cap\ol{L}=\ol{K\cup L}$, $\ol{K}\cup L=\ol{K\setminus L}$, $\ol{K}\cap L=L\setminus K$, $K\cup\ol{L}=\ol{L\setminus K}$, and $\ol{K\oplus L}$.  The witnesses for these functions can be found using the four functions above.

\item
{\bf Product}
The upper bound of $m2^n-2^{n-1}$ was given by Maslov in 1970~\cite{Mas70}, and he stated without proof that it is tight for binary languages 
$$K=\{w\in\{a,b\}^*\mid |w|_a\equiv m-1 \mbox{ mod } m\}$$
and 
$$L=(a^*b)^{n-2}(a\cup b)(b\cup a(a\cup b))^*.$$
The bound was refined by Yu, Zhuang and K.~Salomaa~\cite{YZS94} to $m2^n-k2^{n-1}$, where $k$ is the number of accepting quotients of $K$.
Jir\'asek, Jir\'askov\'a and Szabari~\cite{JJS05} proved that, for any integers $m,n,k$ such that $m\ge2$, $n\ge 2$ and $0<k<m$, there exist binary languages
 $K$ and $L$ such that $\kappa(K)=m$, $\kappa(L)=n$, and $\kappa(K L)=m2^n-k2^{n-1}$. 
Furthermore, Jir\'askov\'a~\cite{Jir09} proved that, for all $m$, $n$, and $\alpha$ such that
either $n=1$ and $1\le \alpha\le m$, or $n\ge 2$ and $1\le\alpha\le m2^n-2^{n-1}$, there exist languages $K$ and $L$ with $\kappa(K)=m$ and $\kappa(L)=n$, defined over a growing alphabet, such that $\kappa(KL)=\alpha$.

For a one-letter alphabet, $mn$ is a tight bound for product if $m$ and $n$ are relatively prime~\cite{YZS94}. The witnesses are \hbox{$K=(a^m)^*a^{m-1}$} and \hbox{$L=(a^n)^*a^{n-1}$}.
See also~\cite{PiSh02}.
\item
{\bf Star}
Maslov~\cite{Mas70} stated\footnote{The bound is incorrectly stated as $2^{n-1}+2^{n-2}-1$, but the example is correct.} without proof that $\kappa(L^*)\le 2^{n-1}+2^{n-2}$, and provided  a binary language meeting this bound.
Three cases were considered by Yu, Zhuang and K.~Salomaa~\cite{YZS94}:
\bi
\item
$n=1$. If $L=\emp$, then $\kappa(L)=1$ and $\kappa(L^*)=2$. If $L=\Sig^*$, then $\kappa(L^*)=1$.
\item
$n=2$. $L=\{w\in \{a,b\}^*\mid |w|_a\equiv 1 \mbox{ mod } 2\}$ has $\kappa(L)=2$, and $\kappa(L^*)=3$.
\item
$n>2$.  Let $\Sig=\{a,b\}$. Then $L=(b\cup a\Sig^{n-1})^*a\Sig^{n-2}$ has $n$ quotients, one of which is accepting, and $\kappa(L^*)=2^{n-1}+2^{n-2}$. This example is different from Maslov's.
\ei
Moreover, Jir\'askov\'a~\cite{Jir08} proved that, for all integers $n$ and $\alpha$ with either $1 = n \le \alpha \le 2$, or $n \ge 2 $
and $1 \le\alpha\le  2^{n-1}+2^{n-2}$, there exists a language $L$ over a $2^n$--letter alphabet such that has $\kappa(L)=n$ and $\kappa(L^*)=\alpha$. 

For a one-letter alphabet, $n^2-2n+2$ is a tight bound for star~\cite{YZS94}. The witness is $L''=(a^n)^*a^{n-1}$. See also~\cite{PiSh02}.
\ei

\section{Generalization of ``non-returning'' state}
\label{sec:non-ret}
A quotient $L_w$ of a language $L$ is {\em \ur\/} if $L_x=L_w$ implies that $x=w$.
If $L_{wa}$ is \ur{} for $a\in \Sig$,  then so is $L_w$. 
Thus, if $L$ has a \ur{} quotient, then $L$ itself is \ur{} by the empty word, \ie, the minimal automaton of $L$ is \emph{non-returning}\footnote{The term ``non-returning'' suggests that once a state is left it cannot be visited again. However,  such non-returning states are not necessarily \ur.}.
Thus the set of \ur{} quotients of $L$ is a tree with root $L$, if it is non-empty.

We now apply the concept of \ur{} quotients to boolean operations and product.
\begin{theorem}
\label{thm:urbool}
Suppose $\kappa(K)=m$, $\kappa(L)=n$,  $K$ and $L$ have $m_u$ and $n_u$ \ur{} quotients, respectively, and there are $r$ words $w_i$ such that both $K_{w_i}$ and~$L_{w_i}$ are \ur{}. 
If $\circ$ is a boolean operator, then
\begin{equation}
\kappa(K\circ L)\leq mn-(\alpha+\beta+\gamma), \mbox{ where}
\end{equation}
\begin{equation}
\alpha = r(m+n)-r(r+1); \;
\beta = (m_u-r)(n-(r+1)); \;
\gamma = (n_u-r)(m-m_u-1).
\end{equation}
If $K$ has $k$ accepting quotients, $t$ of which are \ur, and $s$ rejecting \ur{} quotients, then\begin{equation}
\kappa(KL)\leq m2^n -k2^{n-1} - s(2^n-1) - t(2^{n-1}-1).
\end{equation}

\end{theorem}

The following observation was stated for union and intersection of finite languages in~\cite{Yu01}; we add the suffix-free case:
\begin{corollary}
\label{cor:suffree1}
\quad If $K$ and $L$ are non-empty and finite or suffix-free languages and $\kappa(K)=m>1$,\linebreak $\kappa(L)=n>1$, then $\kappa(K\circ L)\le mn -(m+n-2)$.
\end{corollary}
The bound $mn -(m+n-2)$ for union of suffix-free languages was shown to be tight for quinary languages by Han and Salomaa~\cite{HaSa09}. It is also tight for the binary languages $K=a((ba^*)^{m-3}b)^*(ba^*)^{m-3}$ and $L=a((a\cup b)^{n-3}b)^*(a\cup b)^{n-3}$, as shown recently by Jir\'askov\'a and Olej\'ar~\cite{JiOl09}.
\begin{figure}[ht]
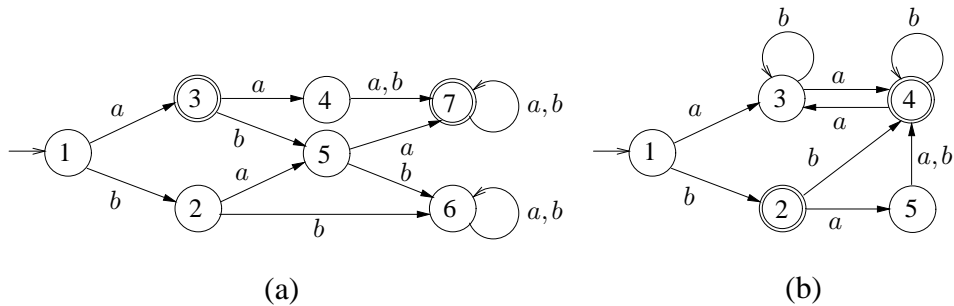


\centerline{\input ur.eepic}

\caption{Illustrating unique reachability.} \label{fig:ur}
\end{figure}

\begin{example}
The automaton of Fig.~\ref{fig:ur}~(a) accepting $K$ has $m=7$ and four \ur{} states: 1, 2, 3, and 4. 
The automaton of Fig.~\ref{fig:ur}~(b) accepting~$L$ has $n=5$ and three \ur{} states: 1, 2, and 5. In pairs $(1,1)$ and~$(2,2)$  both states are reachable by the same word ($\eps$ and $b$, respectively); hence $r=2$.

The $m\times n=7\times 5$ table of all pairs is shown below, where \ur{} states are in boldface type. 
We have $\alpha=18$, where the removed pairs are all the pairs in the first two rows and columns, except $(1,1)$ and $(2,2)$. 
Next, $\beta=4$, and we remove the pairs $(3,4)$, $(3,5)$, $(4,3)$ and $(4,5)$ from rows 3 and 4.
Finally, $\gamma=2$, and we remove the pairs $(6,5)$ and $(7,5)$ from column 5.

{\small
\[
\begin{array}{lllll}
({\bf 1},{\bf 1}) & (1,2) & (1,3) & (1,4) & (1,5)\\
(2,1) & ({\bf 2},{\bf 2}) & (2,3) & (2,4) & (2,5)\\
(3,1) & (3,2) & ({\bf 3},3) & (3,4) & (3,5)\\
(4,1) & (4,2) & (4,3) & ({\bf 4},4) & (4,5)\\
(5,1) & (5,2) & (5,3) & (5,4) & (5,{\bf 5})\\
(6,1) & (6,2) & (6,3) & (6,4) & (6,5)\\
(7,1) & (7,2) & (7,3) & (7,4) & (7,5)
\end{array}
\]
}
Altogether, we have removed 24 states from $K\circ L$, leaving 11 possibilities. The minimal automaton of $K\cup L$ has 8 states.
Notice that state 7 corresponds to the quotient $\Sig^*$. Since $\Sig^*\cup L_w=\Sig^*$ for all $w$, we need to account for only one pair $(7,x)$, and we could remove the remaining four pairs. However, we have already removed pair $(7,5)$ by Theorem~\ref{thm:urbool}.
Hence, there are only three pairs left to remove, and we have an automaton with 8 states. More will be said about the effects of $\Sig^*$ later.

It is also possible to use Theorem~\ref{thm:urbool} if $K$ has some \ur{} quotients and $L$ has none, or when $L$ is completely unknown. If $n_u=0$, then $r=0$, $\alpha=0$, $\beta=m_u(n-1)$, and $\gamma=0$. Then, for any~$L$, 
\begin{equation}
\kappa(K\circ L)\le mn-m_u(n-1).
\end{equation}
For example, for any $L$ with $n=101$  and $K$ as in Fig.~\ref{fig:ur}~(a), $\kappa(K\cap L)\le 307$, instead of the general bound 707. 

Let $K$ and $L$ be the automata of Fig.~\ref{fig:ur}~(a) and~(b), respectively. Then the 
general bound on $\kappa(KL)$ is 192. Here $s=3$ (states 1, 2, and 4), and $t=1$ (state~3).
By Theorem~\ref{thm:urbool} the bound is reduced by $93+15=108$ to 84. The actual 
quotient complexity of $KL$ is 14.

The general bound for $LK$ is 512, the reduced bound is 195, and the actual quotient complexity\linebreak is~12. \hfill\qedb
\end{example}

\section[Languages with special quotients]{Languages with $\eps$, $\Sigma^+$, $\emp$,  or $\Sigma^*$ as quotients}
\label{sec:special}
In this section we consider the effects of the presence of special quotients in a language. In particular, we study the quotients $\eps$, $\Sigma^+$, $\emp$,  and $\Sigma^*$.

\begin{theorem}
\label{thm:bool}
If $\kappa(K)=m$, $\kappa(L)=n$,  and $K$ and $L$ have $k>0$ and $l>0$ accepting quotients, respectively, then
\be
\item
{\it If $K$ and $L$ have $\eps$ as a quotient, then }
	\bi
	\item
	$\kappa(K\cup L)\le mn-2$. 
	\item
	$\kappa(K\cap L)\le mn-(2m+2n-6)$.
	\item
	$\kappa(K\setminus L)\le mn-(m+2n-k-3)$.
	\item
	$\kappa(K\oplus L) \le mn-2$. 
	 \ei
 \item{\it
If $K$ and $L$ have $\Sigma^+$ as a quotient, then }
	\bi
	\item
	$\kappa(K\cap L)\le mn-2$.
	\item
	$\kappa(K\cup L)\le  mn-(2m+2n-6)$.
	\item
	$\kappa(K\setminus L)\le mn-(2m+l-3)$.
	\item
	$\kappa(K\oplus L)\le mn-2$.
	\ei
\item{\it 
If $K$ and $L$ have $\emp$ as a quotient, then}
\bi
	\item
	$\kappa(K\cap L)\le mn-(m+n-2)$.
	\item
	$\kappa(K\setminus L)\le mn-n+1$.
	\ei
\item{\it
If $K$ and $L$ have $\Sigma^*$ as a quotient, then }
\bi
	\item
	$\kappa(K\cup L)\le mn-(m+n-2)$.
	\item
	$\kappa(K\setminus L)\le mn-m+1$.
	\ei
 \item{\it
	\bi
	\item
	If $L$ has $\eps$ as a quotient, then $\kappa(L^R)\le 2^{n-2}+1$.
	\item
	If $L$ has $\Sig^+$ as a quotient, then $\kappa(L^R)\le 2^{n-2}+1$.
	\item
	If $L$ has $\emp$ as a quotient, then $\kappa(L^R)\le 2^{n-1}$.
	\item
	If $L$ has $\Sig^*$ as a quotient, then $\kappa(L^R)\le 2^{n-1}$.
	  \item
	  Moreover,  the effect of these quotients on complexity is cumulative. For example, if $L^R$ has both $\emp$ and $\Sig^*$, then $\kappa(L^R)\le 2^{n-2}$, if $L^R$ has both $\emp$ and $\Sig^+$, then $\kappa(L^R)\le 2^{n-3}+1$, \etc
	\ei}
\ee
\end{theorem}
\goodbreak

\begin{corollary}
If $K$ and $L$ are both non-empty and both suffix-free with $\kappa(K)=m$ and $\kappa(L)=n$, then $\kappa(K\cap L)\leq mn-2(m+n-3)$.
\end{corollary}
It is shown in~\cite{HaSa09} that the bound can be reached with 
$$K=\{\#w \mid w \in \{a,b\}^*,  |w|_a \equiv 0 \mbox{ mod } m-2\},$$ 
$$L=\{\#w \mid w \in \{a,b\}^*,  |w|_b \equiv 0 \mbox{ mod } n-2\}.$$
It was recently proved in~\cite{JiOl09} that this bound can be reached by the binary languages given after Corollary~\ref{cor:suffree1}.

\begin{proposition}
\label{prop:spstar}
If $\kappa(L)=n\ge 3$, $L$ has $l>0$ accepting quotients,  and 
 $L$ has $\eps$ as a quotient,   then $\kappa(L^*)\le 2^{n-3}+2^{n-l-1}+1$.
\end{proposition}
\begin{proof}
If $L$ has $\eps$, then it also has $\emp$.
From~(\ref{eq:star}), every  quotient of $L^*$ by a non-empty word is a union 
of a non-empty subset of quotients of $L$, followed by~$L^*$. 
We have two cases:
\be
\item
Suppose $L$ is rejecting.
	\be
	\item
	If no accepting quotient is included, then there are $2^{n-l-1}-1$ non-empty subsets of non-empty rejecting quotients plus the subset consisting of the empty quotient alone, for a total of $2^{n-l-1}$.
	\item
  	If an accepting quotient is included in the subset, then so is $L$. We can add the subset $\{\eps\}$ or any non-empty subset $S$ of accepting quotients that does not contain $\eps$, since $S\cup \{\eps\}$ is equivalent to $S$.
Thus we have $2^{l-1}$ subsets of accepting quotients.  
To this we can add $2^{n-l-2}$ rejecting subsets, since the  empty quotient and $L$ need not be counted. The total is $2^{l-1}2^{n-l-2}=2^{n-3}$.
	\ee
Adding 1 for $(L^*)_{\eps}$, we have a total of 
$2^{n-3}+2^{n-l-1}+1$.
\item
Suppose $L$ is accepting. Since $n\ge 3$, we have $L\not= \eps$.
	\be
	\item
	If there is no accepting quotient, there are $2^{n-l-1}$ subsets, as before.
	\item If an accepting quotient is included, then $L$ is included and $L$ itself is sufficient to 	guarantee that $(L^*)_w$ is accepting. Since $L\cup \eps=L\cup \emp=L$, we also exclude $\eps$ and $\emp$. Thus any one of  the 
	$2^{n-3}$ 	subsets of the remaining quotients  can be added to $L$.
	\ee
The total is $2^{n-3}+2^{n-l-1}$. We need not add $(L^*)_\eps$, since it is $LL^*$ which has been counted already.
\ee
The worst-case bound of $2^{n-3}+2^{n-l-1}+1$ occurs in the first case only.
\end{proof}

\section{Conclusions}
\label{sec:conc}
Quotients  provide a uniform approach for finding upper bounds for the complexity of operations on regular languages, and for verifying that particular languages meet these bounds. It is hoped that this is a step towards a theory of complexity of languages and automata.

\paragraph{Acknowledgements}
I  am very grateful to Galina Jir\'askov\'a for correcting several errors in early versions of this paper, suggesting better examples, improving proofs, and helping me with references, in particular, with the early work on complexity.
I thank  Sheng Yu for his help with references, and for answering many of my questions on complexity. I also thank Baiyu Li, Shengying Pan, and Jeff Shallit for their careful reading of the manuscript.

\bibliographystyle{eptcs}
\bibliography{brzozo}

\end{document}

%% file: ur.eepic
\setlength{\unitlength}{0.00052493in}
\begingroup\makeatletter\ifx\SetFigFont\undefined%
\gdef\SetFigFont#1#2#3#4#5{%
  \reset@font\fontsize{#1}{#2pt}%
  \fontfamily{#3}\fontseries{#4}\fontshape{#5}%
  \selectfont}%
\fi\endgroup%
{\renewcommand{\dashlinestretch}{30}
\begin{picture}(9443,3087)(0,-10)
\put(4893.929,2049.000){\arc{517.255}{3.8642}{9.0320}}
\path(4781.175,2313.331)(4700.000,2220.000)(4815.545,2264.151)
\path(5907,1597)(6267,1597)
\path(6147.000,1567.000)(6267.000,1597.000)(6147.000,1627.000)
\put(2585,105){\makebox(0,0)[lb]{\smash{{\SetFigFont{12}{14.4}{\rmdefault}{\mddefault}{\updefault}(a)}}}}
\path(12,1597)(372,1597)
\path(252.000,1567.000)(372.000,1597.000)(252.000,1627.000)
\put(7871.000,2542.000){\arc{517.386}{2.2933}{7.4613}}
\path(7606.601,2429.096)(7700.000,2348.000)(7655.752,2463.508)
\put(9176.000,2527.000){\arc{517.386}{2.2933}{7.4613}}
\path(8911.601,2414.096)(9005.000,2333.000)(8960.752,2448.508)
\put(4894.000,969.000){\arc{517.386}{3.8641}{9.0321}}
\path(4781.160,1233.344)(4700.000,1140.000)(4815.538,1184.169)
\put(7842,113){\makebox(0,0)[lb]{\smash{{\SetFigFont{12}{14.4}{\rmdefault}{\mddefault}{\updefault}(b)}}}}
\put(6501,1580){\ellipse{472}{472}}
\put(7815,1017){\ellipse{472}{472}}
\put(9110,2116){\ellipse{472}{472}}
\put(9127,1010){\ellipse{472}{472}}
\put(7812,1019){\ellipse{406}{406}}
\put(9114,2114){\ellipse{406}{406}}
\put(1910,2130){\ellipse{406}{406}}
\put(3215,2116){\ellipse{472}{472}}
\put(3215,1574){\ellipse{472}{472}}
\path(3453,1619)(4368,1904)
\blacken\path(4262.351,1839.671)(4368.000,1904.000)(4244.508,1896.957)(4262.351,1839.671)
\path(2135,960)(4250,960)
\blacken\path(4130.000,930.000)(4250.000,960.000)(4130.000,990.000)(4130.000,930.000)
\path(3463,2130)(4280,2130)
\blacken\path(4160.000,2100.000)(4280.000,2130.000)(4160.000,2160.000)(4160.000,2100.000)
\path(800,1425)(1692,1088)
\blacken\path(1569.142,1102.347)(1692.000,1088.000)(1590.347,1158.474)(1569.142,1102.347)
\path(2110,1987)(3002,1650)
\blacken\path(2879.142,1664.347)(3002.000,1650.000)(2900.347,1720.474)(2879.142,1664.347)
\path(2158,2130)(2975,2130)
\blacken\path(2855.000,2100.000)(2975.000,2130.000)(2855.000,2160.000)(2855.000,2100.000)
\path(6695,1425)(7587,1088)
\blacken\path(7464.142,1102.347)(7587.000,1088.000)(7485.347,1158.474)(7464.142,1102.347)
\path(8030,2220)(8885,2220)
\blacken\path(8765.000,2190.000)(8885.000,2220.000)(8765.000,2250.000)(8765.000,2190.000)
\path(8885,2040)(8030,2040)
\blacken\path(8150.000,2070.000)(8030.000,2040.000)(8150.000,2010.000)(8150.000,2070.000)
\path(8030,1155)(9005,1898)
\blacken\path(8927.738,1801.405)(9005.000,1898.000)(8891.371,1849.127)(8927.738,1801.405)
\path(8060,1013)(8893,1013)
\blacken\path(8773.000,983.000)(8893.000,1013.000)(8773.000,1043.000)(8773.000,983.000)
\path(9118,1260)(9118,1875)
\blacken\path(9148.000,1755.000)(9118.000,1875.000)(9088.000,1755.000)(9148.000,1755.000)
\path(6725,1680)(7580,2085)
\blacken\path(7484.394,2006.518)(7580.000,2085.000)(7458.709,2060.742)(7484.394,2006.518)
\path(830,1680)(1685,2085)
\blacken\path(1589.394,2006.518)(1685.000,2085.000)(1563.709,2060.742)(1589.394,2006.518)
\path(2150,1087)(2990,1485)
\blacken\path(2894.402,1406.508)(2990.000,1485.000)(2868.711,1460.729)(2894.402,1406.508)
\path(3433,1470)(4288,1132)
\blacken\path(4165.375,1148.217)(4288.000,1132.000)(4187.433,1204.016)(4165.375,1148.217)
\put(5240,1995){\makebox(0,0)[lb]{\smash{{\SetFigFont{10}{12.0}{\rmdefault}{\mddefault}{\updefault}$a,b$}}}}
\put(5240,915){\makebox(0,0)[lb]{\smash{{\SetFigFont{10}{12.0}{\rmdefault}{\mddefault}{\updefault}$a,b$}}}}
\put(3080,735){\makebox(0,0)[lb]{\smash{{\SetFigFont{10}{12.0}{\rmdefault}{\mddefault}{\updefault}$b$}}}}
\put(2270,1635){\makebox(0,0)[lb]{\smash{{\SetFigFont{10}{12.0}{\rmdefault}{\mddefault}{\updefault}$b$}}}}
\put(6411,1507){\makebox(0,0)[lb]{\smash{{\SetFigFont{10}{12.0}{\rmdefault}{\mddefault}{\updefault}$1$}}}}
\put(6860,1950){\makebox(0,0)[lb]{\smash{{\SetFigFont{10}{12.0}{\rmdefault}{\mddefault}{\updefault}$a$}}}}
\put(6815,1050){\makebox(0,0)[lb]{\smash{{\SetFigFont{10}{12.0}{\rmdefault}{\mddefault}{\updefault}$b$}}}}
\put(8307,2280){\makebox(0,0)[lb]{\smash{{\SetFigFont{10}{12.0}{\rmdefault}{\mddefault}{\updefault}$a$}}}}
\put(8322,1822){\makebox(0,0)[lb]{\smash{{\SetFigFont{10}{12.0}{\rmdefault}{\mddefault}{\updefault}$a$}}}}
\put(8067,1447){\makebox(0,0)[lb]{\smash{{\SetFigFont{10}{12.0}{\rmdefault}{\mddefault}{\updefault}$b$}}}}
\put(9170,1469){\makebox(0,0)[lb]{\smash{{\SetFigFont{10}{12.0}{\rmdefault}{\mddefault}{\updefault}$a,b$}}}}
\put(8277,810){\makebox(0,0)[lb]{\smash{{\SetFigFont{10}{12.0}{\rmdefault}{\mddefault}{\updefault}$a$}}}}
\put(3627,2220){\makebox(0,0)[lb]{\smash{{\SetFigFont{10}{12.0}{\rmdefault}{\mddefault}{\updefault}$a,b$}}}}
\put(2292,1305){\makebox(0,0)[lb]{\smash{{\SetFigFont{10}{12.0}{\rmdefault}{\mddefault}{\updefault}$a$}}}}
\put(3957,1298){\makebox(0,0)[lb]{\smash{{\SetFigFont{10}{12.0}{\rmdefault}{\mddefault}{\updefault}$b$}}}}
\put(3957,1560){\makebox(0,0)[lb]{\smash{{\SetFigFont{10}{12.0}{\rmdefault}{\mddefault}{\updefault}$a$}}}}
\put(7738,921){\makebox(0,0)[lb]{\smash{{\SetFigFont{10}{12.0}{\rmdefault}{\mddefault}{\updefault}$2$}}}}
\put(9027,2033){\makebox(0,0)[lb]{\smash{{\SetFigFont{10}{12.0}{\rmdefault}{\mddefault}{\updefault}$4$}}}}
\put(4392,930){\makebox(0,0)[lb]{\smash{{\SetFigFont{10}{12.0}{\rmdefault}{\mddefault}{\updefault}$6$}}}}
\put(4392,1987){\makebox(0,0)[lb]{\smash{{\SetFigFont{10}{12.0}{\rmdefault}{\mddefault}{\updefault}$7$}}}}
\put(1828,936){\makebox(0,0)[lb]{\smash{{\SetFigFont{10}{12.0}{\rmdefault}{\mddefault}{\updefault}$2$}}}}
\put(1821,2040){\makebox(0,0)[lb]{\smash{{\SetFigFont{10}{12.0}{\rmdefault}{\mddefault}{\updefault}$3$}}}}
\put(509,1492){\makebox(0,0)[lb]{\smash{{\SetFigFont{10}{12.0}{\rmdefault}{\mddefault}{\updefault}$1$}}}}
\put(3132,2025){\makebox(0,0)[lb]{\smash{{\SetFigFont{10}{12.0}{\rmdefault}{\mddefault}{\updefault}$4$}}}}
\put(3125,1478){\makebox(0,0)[lb]{\smash{{\SetFigFont{10}{12.0}{\rmdefault}{\mddefault}{\updefault}$5$}}}}
\put(7723,2040){\makebox(0,0)[lb]{\smash{{\SetFigFont{10}{12.0}{\rmdefault}{\mddefault}{\updefault}$3$}}}}
\put(9049,931){\makebox(0,0)[lb]{\smash{{\SetFigFont{10}{12.0}{\rmdefault}{\mddefault}{\updefault}$5$}}}}
\put(2457,2197){\makebox(0,0)[lb]{\smash{{\SetFigFont{10}{12.0}{\rmdefault}{\mddefault}{\updefault}$a$}}}}
\put(1033,1957){\makebox(0,0)[lb]{\smash{{\SetFigFont{10}{12.0}{\rmdefault}{\mddefault}{\updefault}$a$}}}}
\put(1025,1019){\makebox(0,0)[lb]{\smash{{\SetFigFont{10}{12.0}{\rmdefault}{\mddefault}{\updefault}$b$}}}}
\put(7767,2865){\makebox(0,0)[lb]{\smash{{\SetFigFont{10}{12.0}{\rmdefault}{\mddefault}{\updefault}$b$}}}}
\put(9080,2857){\makebox(0,0)[lb]{\smash{{\SetFigFont{10}{12.0}{\rmdefault}{\mddefault}{\updefault}$b$}}}}
\put(4487,1017){\ellipse{472}{472}}
\put(1910,2130){\ellipse{472}{472}}
\put(606,1580){\ellipse{472}{472}}
\put(1920,1017){\ellipse{472}{472}}
\put(4486,2095){\ellipse{472}{472}}
\put(4483,2092){\ellipse{406}{406}}
\put(7805,2130){\ellipse{472}{472}}
\end{picture}
}